\documentclass[a4paper,twocolumn,10pt]{article}
\usepackage[utf8]{inputenc}
\usepackage{mathptmx} 
\usepackage[hyphens]{url}
\usepackage{doi}
\usepackage{hyperref}
\usepackage[hyphenbreaks]{breakurl} 
\usepackage[numbers,sort]{natbib}
\usepackage{amssymb}
\usepackage{amsthm} 
\newtheorem{theorem}{Assertion}
\newtheorem{corollary}[theorem]{Corollary}
\begin{document}
\sloppy
\date{} 
\title{A Critique of the CAP Theorem}
\author{Martin Kleppmann}
\maketitle

\subsection*{Abstract}

The \emph{CAP Theorem} is a frequently cited impossibility result in distributed systems, especially
among \emph{NoSQL} distributed databases. In this paper we survey some of the confusion about the
meaning of CAP, including inconsistencies and ambiguities in its definitions, and we highlight some
problems in its formalization. CAP is often interpreted as proof that eventually consistent
databases have better availability properties than strongly consistent databases; although there is
some truth in this, we show that more careful reasoning is required. These problems cast doubt on
the utility of CAP as a tool for reasoning about trade-offs in practical systems. As alternative to
CAP, we propose a \emph{delay-sensitivity} framework, which analyzes the sensitivity of operation
latency to network delay, and which may help practitioners reason about the trade-offs between
consistency guarantees and tolerance of network faults.

\section{Background}

Replicated databases maintain copies of the same data on multiple nodes, potentially in disparate
geographical locations, in order to tolerate faults (failures of nodes or communication links) and
to provide lower latency to users (requests can be served by a nearby site). However, implementing
reliable, fault-tolerant applications in a distributed system is difficult: if there are multiple
copies of the data on different nodes, they may be inconsistent with each other, and an application
that is not designed to handle such inconsistencies may produce incorrect results.

In order to provide a simpler programming model to application developers, the designers of
distributed data systems have explored various consistency guarantees that can be implemented by the
database infrastructure, such as linearizability~\cite{Herlihy1990jq}, sequential
consistency~\cite{Lamport1979ky}, causal consistency~\cite{Ahamad1995gl} and pipelined RAM
(PRAM)~\cite{Lipton1988uh}. When multiple processes execute operations on a shared storage
abstraction such as a database, a consistency model describes what values are allowed to be returned
by operations accessing the storage, depending on other operations executed previously or
concurrently, and the return values of those operations.

Similar concerns arise in the design of multiprocessor computers, which are not geographically
distributed, but nevertheless present inconsistent views of memory to different threads, due to the
various caches and buffers employed by modern CPU architectures. For example, x86 microprocessors
provide a level of consistency that is weaker than sequential, but stronger than causal
consistency~\cite{Sewell2010fj}. However, in this paper we focus our attention on distributed
systems that must tolerate partial failures and unreliable network links.

A strong consistency model like linearizability provides an easy-to-understand guarantee:
informally, all operations behave as if they executed atomically on a \emph{single copy} of the
data. However, this guarantee comes at the cost of reduced performance~\cite{Attiya1994gw} and fault
tolerance~\cite{Davidson1985hv} compared to weaker consistency models. In particular, as we discuss
in this paper, algorithms that ensure stronger consistency properties among replicas are more
sensitive to message delays and faults in the network. Many real computer networks are prone to
unbounded delays and lost messages~\cite{Bailis2014jx}, making the fault tolerance of distributed
consistency algorithms an important issue in practice.

A \emph{network partition} is a particular kind of communication fault that splits the network into
subsets of nodes such that nodes in one subset cannot communicate with nodes in another. As long as
the partition exists, any data modifications made in one subset of nodes cannot be visible to nodes
in another subset, since all messages between them are lost. Thus, an algorithm that maintains the
illusion of a single copy may have to delay operations until the partition is healed, to avoid the
risk of introducing inconsistent data in different subsets of nodes.

This trade-off was already known in the
1970s~\cite{Johnson1975we, Lindsay1979wv, Fischer1982hc, Davidson1985hv},
but it was rediscovered in the early
2000s, when the web's growing commercial popularity made geographic distribution and high
availability important to many organizations~\cite{Brewer2012tr, Vogels2008ey}. It was originally
called the \emph{CAP Principle} by Fox and Brewer~\cite{Fox1999bs, Brewer2000vd}, where CAP stands
for \emph{Consistency}, \emph{Availability} and \emph{Partition tolerance}. After the principle was
formalized by Gilbert and Lynch~\cite{Gilbert2002il, Gilbert2012bf} it became known as the
\emph{CAP Theorem}.

CAP became an influential idea in the NoSQL movement~\cite{Vogels2008ey}, and was adopted by
distributed systems practitioners to critique design decisions~\cite{Hodges2013tj}. It provoked a
lively debate about trade-offs in data systems, and encouraged system designers to challenge the
received wisdom that strong consistency guarantees were essential for databases~\cite{Brewer2012ba}.

The rest of this paper is organized as follows: in section~\ref{sec:definitions} we compare various
definitions of consistency, availability and partition tolerance. We then examine the formalization
of CAP by Gilbert and Lynch~\cite{Gilbert2002il} in section~\ref{sec:proofs}. Finally, in
section~\ref{sec:alternatives} we discuss some alternatives to CAP that are useful for reasoning
about trade-offs in distributed systems.

\section{CAP Theorem Definitions}\label{sec:definitions}

CAP was originally presented in the form of \emph{``consistency, availability, partition tolerance:
pick any two''} (i.e.\ you can have CA, CP or AP, but not all three). Subsequent debates concluded
that this formulation is misleading~\cite{Brewer2012ba, Hale2010we, Robinson2010tp}, because the
distinction between CA and CP is unclear, as detailed later in this section. Many authors now prefer
the following formulation: if there is no network partition, a system can be both consistent and
available; when a network partition occurs, a system must choose between either consistency (CP) or
availability (AP).

Some authors~\cite{Darcy2010ta, Liochon2015vt} define a CP system as one in which a majority of
nodes on one side of a partition can continue operating normally, and a CA system as one that may
fail catastrophically under a network partition (since it is designed on the assumption that
partitions are very rare). However, this definition is not universally agreed, since it is
counter-intuitive to label a system as ``available'' if it fails catastrophically under a partition,
while a system that continues partially operating in the same situation is labelled ``unavailable''
(see section~\ref{sec:partitions}).

Disagreement about the definitions of terms like \emph{availability} is the source of many
misunderstandings about CAP, and unclear definitions lead to problems with its formalization as a
theorem. In sections~\ref{sec:availability} to \ref{sec:partitions} we survey various definitions
that have been proposed.

\subsection{Availability}\label{sec:availability}

In practical engineering terms, \emph{availability} usually refers to the proportion of time during
which a service is able to successfully handle requests, or the proportion of requests that receive
a successful response. A response is usually considered successful if it is valid (not an error, and
satisfies the database's safety properties) and it arrives at the client within some timeout, which
may be specified in a \emph{service level agreement} (SLA). Availability in this sense is a metric
that is empirically observed during a period of a service's operation. A service may be available
(up) or unavailable (down) at any given time, but it is nonsensical to say that some software
package or algorithm is `available' or `unavailable' in general, since the uptime percentage is only
known in retrospect, after a period of operation (during which various faults may have occurred).

There is a long tradition of \emph{highly available} and \emph{fault-tolerant} systems, whose
algorithms are designed such that the system can remain available (up) even when some part of the
system is faulty, thus increasing the expected mean time to failure (MTTF) of the system as a whole.
Using such a system does not automatically make a service 100\% available, but it may increase the
observed availability during operation, compared to using a system that is not fault-tolerant.

\subsubsection{The A in CAP}\label{sec:a-in-cap}

Does the A in CAP refer to a property of an algorithm, or to an observed metric during system
operation? This distinction is unclear. Brewer does not offer a precise definition of availability,
but states that ``availability is obviously continuous from 0 to 100 percent''~\cite{Brewer2012ba},
suggesting an observed metric. Fox and Brewer also use the term \emph{yield} to refer to the
proportion of requests that are completed successfully~\cite{Fox1999bs} (without specifying any
timeout).

On the other hand, Gilbert and Lynch~\cite{Gilbert2002il} write: ``For a distributed system to be
continuously available, every request received by a non-failing node in the system must result in a
response''.\footnote{This sentence appears to define a property of \emph{continuous availability},
but the rest of the paper does not refer to this ``continuous'' aspect.} In order to prove a result
about systems in general, this definition interprets availability as a property of an algorithm, not
as an observed metric during system operation -- i.e.\ they define a system as being ``available''
or ``unavailable'' statically, based on its algorithms, not its operational status at some point in
time.

One particular execution of the algorithm is available if every request in that execution eventually
receives a response. Thus, an algorithm is ``available'' under Gilbert and Lynch's definition if
\emph{all} possible executions of the algorithm are available. That is, the algorithm must guarantee
that requests always result in responses, no matter what happens in the system (see
section~\ref{sec:systemmodel}).

Note that Gilbert and Lynch's definition requires \emph{any} non-failed node to be able to generate
valid responses, even if that node is completely isolated from the other nodes. This definition is
at odds with Fox and Brewer's original proposal of CAP, which states that ``data is considered
highly available if a given consumer of the data can always reach \emph{some}
replica''~\cite[emphasis original]{Fox1999bs}.

Many so-called highly available or fault-tolerant systems have very high uptime in practice, but are
in fact ``unavailable'' under Gilbert and Lynch's definition~\cite{Kim1984bl}: for example, in a
system with an elected leader or primary node, if a client that cannot reach the leader due to a
network fault, the client cannot perform any writes, even though it may be able to reach another
replica.

\subsubsection{No maximum latency}\label{sec:no-max-latency}

Note that Gilbert and Lynch's definition of availability does not specify any upper bound on
operation latency: it only requires requests to \emph{eventually} return a response within some
unbounded but finite time. This is convenient for proof purposes, but does not closely match our
intuitive notion of availability (in most situations, a service that takes a week to respond might
as well be considered unavailable).

This definition of availability is a pure liveness property, not a safety
property~\cite{Alpern1985dg}: that is, at any point in time, if the response has not yet arrived,
there is still hope that the availability property might still be fulfilled, because the response
may yet arrive -- it is never too late. This aspect of the definition will be important in
section~\ref{sec:proofs}, when we examine Gilbert and Lynch's proofs in more detail.

(In section~\ref{sec:alternatives} we will discuss a definition of availability that takes latency
into account.)

\subsubsection{Failed nodes}\label{sec:failed-node-exception}

Another noteworthy aspect of Gilbert and Lynch's definition of availability is the proviso of
applying only to \emph{non-failed} nodes. This allows the aforementioned definition of a CA system
as one that fails catastrophically if a network partition occurs: if the partition causes all nodes
to fail, then the availability requirement does not apply to any nodes, and thus it is trivially
satisfied, even if no node is able to respond to any requests. This definition is logically sound,
but somewhat counter-intuitive.

\subsection{Consistency}\label{sec:consistency}

Consistency is also an overloaded word in data systems: consistency in the sense of ACID is a very
different property from consistency in CAP~\cite{Brewer2012ba}. In the distributed systems
literature, consistency is usually understood as not one particular property, but as a spectrum of
models with varying strengths of guarantee. Examples of such consistency models include
linearizability~\cite{Herlihy1990jq}, sequential consistency~\cite{Lamport1979ky}, causal
consistency~\cite{Ahamad1995gl} and PRAM~\cite{Lipton1988uh}.

There is some similarity between consistency models and \emph{transaction isolation models} such as
serializability~\cite{Bernstein1987va}, snapshot isolation~\cite{Berenson1995kj}, repeatable read
and read committed~\cite{Adya1999tx, Gray1976us}. Both describe restrictions on the values that
operations may return, depending on other (prior or concurrent) operations. The difference is that
transaction isolation models are usually formalized assuming a single replica, and operate at the
granularity of transactions (each transaction may read or write multiple objects). Consistency
models assume multiple replicas, but are usually defined in terms of single-object operations (not
grouped into transactions). Bailis et al.~\cite{Bailis2014vc} demonstrate a unified framework for
reasoning about both distributed consistency and transaction isolation in terms of CAP.

\subsubsection{The C in CAP}\label{sec:c-in-cap}

Fox and Brewer~\cite{Fox1999bs} define the C in CAP as one-copy serializability
(1SR)~\cite{Bernstein1987va}, whereas Gilbert and Lynch~\cite{Gilbert2002il} define it as
linearizability. Those definitions are not identical, but fairly similar.\footnote{Linearizability
is a recency guarantee, whereas 1SR is not. 1SR requires isolated execution of multi-object
transactions, which linearizability does not. Both require \emph{coordination}, in the sense of
section~\ref{sec:disconnected}~\cite{Bailis2014vc}.} Both are safety properties~\cite{Alpern1985dg},
i.e.\ restrictions on the possible executions of the system, ensuring that certain situations never
occur.

In the case of linearizability, the situation that may not occur is a \emph{stale read}: stated
informally, once a write operation has completed or some read operation has returned a new value,
all following read operations must return the new value, until it is overwritten by another write
operation. Gilbert and Lynch observe that if the write and read operations occur on different nodes,
and those nodes cannot communicate during the time when those operations are being executed, then
the safety property cannot be satisfied, because the read operation cannot know about the value
written.

The C of CAP is sometimes referred to as \emph{strong consistency} (a term that is not formally
defined), and contrasted with \emph{eventual consistency}~\cite{Terry1994fp, Vogels2008ey, Bailis2013jc},
which is often regarded as the weakest level of consistency that is useful to applications. Eventual
consistency means that if a system stops accepting writes and sufficient\footnote{It is not clear
what amount of communication is `sufficient'. A possible formalization would be to require all
replicas to converge to the same value within finite time, assuming fair-loss links (see section
\ref{sec:fairloss}).} communication occurs, then eventually all replicas will converge to the same
value. However, as the aforementioned list of consistency models indicates, it is overly simplistic
to cast `strong' and eventual consistency as the only possible choices.

\subsubsection{Probabilistic consistency}

It is also possible to define consistency as a quantitative metric rather than a safety property.
For example, Fox and Brewer~\cite{Fox1999bs} define \emph{harvest} as ``the fraction of the data
reflected in the response, i.e.\ the completeness of the answer to the query,'' and
probabilistically bounded staleness~\cite{Bailis2012to} studies the probability of a read returning
a stale value, given various assumptions about the distribution of network latencies. However, these
stochastic definitions of consistency are not the subject of CAP.

\subsection{Partition Tolerance}\label{sec:partitions}

A \emph{network partition} has long been defined as a communication failure in which the network is
split into disjoint sub-networks, with no communication possible across
sub-networks~\cite{Johnson1975we}. This is a fairly narrow class of fault, but it does occur in
practice~\cite{Bailis2014jx}, so it is worth studying.

\subsubsection{Assumptions about system model}\label{sec:systemmodel}

It is less clear what \emph{partition tolerance} means. Gilbert and Lynch~\cite{Gilbert2002il}
define a system as partition-tolerant if it continues to satisfy the consistency and availability
properties in the presence of a partition. Fox and Brewer~\cite{Fox1999bs} define
\emph{partition-resilience} as ``the system as whole can survive a partition between data replicas''
(where \emph{survive} is not defined).

At first glance, these definitions may seem redundant: if we say that an algorithm provides some
guarantee (e.g. linearizability), then we expect \emph{all} executions of the algorithm to satisfy
that property, regardless of the faults that occur during the execution.

However, we can clarify the definitions by observing that the correctness of a distributed algorithm
is always subject to assumptions about the faults that may occur during its execution. If you take
an algorithm that assumes fair-loss links and crash-stop processes, and subject it to Byzantine
faults, the execution will most likely violate safety properties that were supposedly guaranteed.
These assumptions are typically encoded in a \emph{system model}, and non-Byzantine system models
rule out certain kinds of fault as impossible (so algorithms are not expected to tolerate them).

Thus, we can interpret \emph{partition tolerance} as meaning ``a network partition is among the
faults that are assumed to be possible in the system.'' Note that this definition of partition
tolerance is a statement about the system model, whereas consistency and availability are properties
of the possible executions of an algorithm. It is misleading to say that an algorithm ``provides
partition tolerance,'' and it is better to say that an algorithm ``assumes that partitions may
occur.''

If an algorithm assumes the absence of partitions, and is nevertheless subjected to a partition, it
may violate its guarantees in arbitrarily undefined ways (including failing to respond even after
the partition is healed, or deleting arbitrary amounts of data). Even though it may seem that such
arbitrary failure semantics are not very useful, various systems exhibit such behavior in
practice~\cite{Kingsbury2014tk, Kingsbury2015uk}. Making networks highly reliable is very
expensive~\cite{Bailis2014jx}, so most distributed programs must assume that partitions will occur
sooner or later~\cite{Hale2010we}.

\subsubsection{Partitions and fair-loss links}\label{sec:fairloss}

Further confusion arises due to the fact that network partitions are only one of a wide range of
faults that can occur in distributed systems, including nodes failing or restarting, nodes pausing
for some amount of time (e.g.\ due to garbage collection), and loss or delay of messages in the
network. Some faults can be modeled in terms of other faults (for example, Gilbert and Lynch state
that the loss of an individual message can be modeled as a short-lived network partition).

In the design of distributed systems algorithms, a commonly assumed system model is \emph{fair-loss
links}~\cite{Cachin2011wt}. A network link has the fair-loss property if the probability of a
message \emph{not} being lost is non-zero, i.e.\ the link sometimes delivers messages. The link may
have intervals of time during which all messages are dropped, but those intervals must be of finite
duration. On a fair-loss link, message delivery can be made reliable by retrying a message an
unbounded number of times: the message is guaranteed to be eventually delivered after a finite
number of attempts~\cite{Cachin2011wt}.

We argue that fair-loss links are a good model of most networks in practice: faults occur
unpredictably; messages are lost while the fault is occurring; the fault lasts for some finite
duration (perhaps seconds, perhaps hours), and eventually it is healed (perhaps after human
intervention). There is no malicious actor in the network who can cause systematic message loss over
unlimited periods of time -- such malicious actors are usually only assumed in the design of
Byzantine fault tolerant algorithms.

Is ``partitions may occur'' equivalent to assuming fair-loss links? Gilbert and
Lynch~\cite{Gilbert2002il} define partitions as ``the network will be allowed to lose arbitrarily
many messages sent from one node to another.'' In this definition it is unclear whether the number
of lost messages is unbounded but finite, or whether it is potentially infinite.

Partitions of a finite duration are possible with fair-loss links, and thus an algorithm that is
correct in a system model of fair-loss links can tolerate partitions of a finite duration.
Partitions of an infinite duration require some further thought, as we shall see in
section~\ref{sec:proofs}.

\section{The CAP Proofs}\label{sec:proofs}

In this section, we build upon the discussion of definitions in the last section, and examine the
proofs of the theorems of Gilbert and Lynch~\cite{Gilbert2002il}. We highlight some ambiguities in
the reasoning of the proofs, and then suggest a more precise formalization.

\subsection{Theorems 1 and 2}\label{sec:theorem1}

Gilbert and Lynch's Theorem 1 is stated as follows:

\emph{It is impossible in the asynchronous network model to implement a read/write data object that
guarantees the following properties:}
\begin{itemize}
\item \emph{Availability}
\item \emph{Atomic consistency}\footnote{In this context, \emph{atomic consistency} is synonymous
with linearizability, and it is unrelated to the A in ACID.}
\end{itemize}
\emph{in all fair executions (including those in which messages are lost).}

Theorem 2 is similar, but specified in a system model with bounded network delay. The discussion in
this section~\ref{sec:theorem1} applies to both theorems.

\subsubsection{Availability of failed nodes}\label{sec:failed-node-trivial}

The first problem with this proof is the definition of availability. As discussed in
section~\ref{sec:failed-node-exception}, only non-failing nodes are required to respond.

If it is possible for the algorithm to declare nodes as failed (e.g.\ if a node may crash itself),
then the availability property can be trivially satisfied: all nodes can be crashed, and thus no
node is required to respond. Of course, such an algorithm would not be useful in practice.
Alternatively, if a minority of nodes is permanently partitioned from the majority, an algorithm
could define the nodes in the minority partition as failed (by crashing them), while the majority
partition continues implementing a linearizable register~\cite{Attiya1995bm}.

This is not the intention of CAP -- the raison d'\^{e}tre of CAP is to characterize systems in which
a minority partition can continue operating independently of the rest -- but the present
formalization of availability does not exclude such trivial solutions.

\subsubsection{Finite and infinite partitions}

Gilbert and Lynch's proofs of theorems 1 and 2 construct an execution of an algorithm $A$ in which a
write is followed by a read, while simultaneously a partition exists in the network. By showing that
the execution is not linearizable, the authors derive a contradiction.

Note that this reasoning is only correct if we assume a system model in which partitions may have
infinite duration.

If the system model is based on fair-loss links, then all partitions may be assumed to be of
unbounded but finite duration (section~\ref{sec:fairloss}). Likewise, Gilbert and Lynch's
availability property does not place any upper bound on the duration of an operation, as long as it
is finite (section~\ref{sec:no-max-latency}). Thus, if a linearizable algorithm encounters a
network partition in a fair-loss system model, it is acceptable for the algorithm to simply wait for
the partition to be healed: at any point in time, there is still hope that the partition will be
healed in future, and so the availability property may yet be satisfied. For example, the ABD
algorithm~\cite{Attiya1995bm} can be used to implement a linearizable read-write register in an
asynchronous network with fair-loss links.\footnote{ABD~\cite{Attiya1995bm} is an algorithm for a
single-writer multi-reader register. It was extended to the multi-writer case by Lynch and
Shvartsman~\cite{Lynch1997gr}.}

On the other hand, in an execution where a partition of infinite duration occurs, the algorithm is
forced to make a choice between waiting until the partition heals (which never happens, thus
violating availability) and exhibiting the execution in the proof of Theorem 1 (thus violating
linearizability). We can conclude that Theorem 1 is only valid in a system model where infinite
partitions are possible.

\subsubsection{Linearizability vs. eventual consistency}

Note that in the case of an infinite partition, no information can ever flow from one sub-network to
the other. Thus, even eventual consistency (replica convergence in finite time, see
section~\ref{sec:c-in-cap}) is not possible in a system with an infinite partition.

Theorem 1 demonstrated that in a system model with infinite partitions, no algorithm exists which
ensures linearizability and availability in all executions. However, we can also see that in the
same system model, no algorithm exists which ensures eventual consistency in all executions.

The CAP theorem is often understood as demonstrating that linearizability cannot be achieved with
high availability, whereas eventual consistency can. However, the results so far do not
differentiate between linearizable and eventually consistent algorithms: both are possible if
partitions are always finite, and both are impossible in a system model with infinite partitions.

To distinguish between linearizability and eventual consistency, a more careful formalization of CAP
is required, which we give in section~\ref{sec:partitionable-model}.

\subsection{The partitionable system model}\label{sec:partitionable-model}

In this section we suggest a more precise formulation of CAP, and derive a result similar to Gilbert
and Lynch's Theorem 1 and Corollary 1.1. This formulation will help us gain a better understanding
of CAP and its consequences.

\subsubsection{Definitions}

Define a \emph{partitionable link} to be a point-to-point link with the following properties:
\begin{enumerate}
    \item \emph{No duplication:} If a process $p$ sends a message $m$ once to process $q$, then $m$
        is delivered at most once by $q$.
    \item \emph{No creation:} If some process $q$ delivers a message $m$ with sender $p$, then $m$
        was previously sent to $q$ by process $p$.
\end{enumerate}
(A partitionable link is allowed to drop an infinite number of messages and cause unbounded message
delay.)

Define the \emph{partitionable model} as a system model in which processes can only communicate via
partitionable links, in which processes never crash,\footnote{The assumption that processes never
crash is of course unrealistic, but it makes the impossibility results in
section~\ref{sec:impossibility} stronger. It also rules out the trivial solution of
section~\ref{sec:failed-node-trivial}.} and in which every process has access to a local clock that
is able to generate timeouts (the clock progresses monotonically at a rate approximately equal to
real time, but clocks of different processes are not synchronized).

Define an execution $E$ as \emph{admissible} in a system model $M$ if the processes and links in $E$
satisfy the properties defined by $M$.

Define an algorithm $A$ as \emph{terminating} in a system model $M$ if, for every execution $E$ of
$A$, if $E$ is admissible in $M$, then every operation in $E$ terminates in finite
time.\footnote{Our definition of \emph{terminating} corresponds to Gilbert and Lynch's definition of
\emph{available}. We prefer to call it \emph{terminating} because the word \emph{available} is
widely understood as referring to an empirical metric (see section~\ref{sec:availability}). There is
some similarity to \emph{wait-free} data structures~\cite{Herlihy1988br}, although these usually
assume reliable communication and unreliable processes.}

Define an execution $E$ as \emph{loss-free} if for every message $m$ sent from $p$ to $q$ during
$E$, $m$ is eventually delivered to $q$. (There is no delay bound on delivery.)

An execution $E$ is \emph{partitioned} if it is not loss-free. Note: we may assume that links
automatically resend lost messages an unbounded number of times. Thus, an execution in which
messages are transiently lost during some finite time period is not partitioned, because the links
will eventually deliver all messages that were lost. An execution is only partitioned if the message
loss persists forever.

For the definition of \emph{linearizability} we refer to Herlihy and Wing~\cite{Herlihy1990jq}.

There is no generally agreed formalization of \emph{eventual consistency}, but the following
corresponds to a liveness property that has been proposed~\cite{Bailis2013jc, Attiya2015dm}:
eventually, every read operation $\mathit{read}(q)$ at process $q$ must return a set of all the
values $v$ ever written by any process $p$ in an operation $\mathit{write}(p, v)$. For simplicity,
we assume that values are never removed from the read set, although an application may only see one
of the values (e.g.\ the one with the highest timestamp).

More formally, an infinite execution $E$ is \emph{eventually consistent} if, for all processes $p$
and $q$, and for every value $v$ such that operation $\mathit{write}(p, v)$ occurs in $E$, there are
only finitely many operations in $E$ such that $v \notin \mathit{read}(q)$.

\subsubsection{Impossibility results}\label{sec:impossibility}

\begin{theorem}\label{th:1}
If an algorithm $A$ implements a terminating read-write register $R$ in the partitionable model,
then there exists a loss-free execution of $A$ in which $R$ is not linearizable.
\end{theorem}

\begin{proof}
Consider an execution $E_1$ in which the initial value of $R$ is $v_1$, and no messages are
delivered (all messages are lost, which is admissible for partitionable links). In $E_1$, $p$ first
performs an operation $\mathit{write}(p, v_2)$ where $v_2 \neq v_1$. This operation must terminate
in finite time due to the termination property of $A$.

After the write operation terminates, $q$ performs an operation $\mathit{read}(q)$, which must
return $v_1$, since there is no way for $q$ to know the value $v_2$, due to all messages being lost.
The read must also terminate. This execution is not linearizable, because the read did not return
$v_2$.

Now consider an execution $E_2$ which extends $E_1$ as follows: after the termination of the
$\mathit{read}(q)$ operation, every message that was sent during $E_1$ is delivered (this is
admissible for partitionable links). These deliveries cannot affect the execution of the write and
read operations, since they occur after the termination of both operations, so $E_2$ is also
non-linearizable. Moreover, $E_2$ is loss-free, since every message was delivered.
\end{proof}

\begin{corollary}\label{corr:2}
There is no algorithm that implements a terminating read-write register in the partitionable model
that is linearizable in all loss-free executions.
\end{corollary}

This corresponds to Gilbert and Lynch's Corollary~1, and follows directly from the existence of a
loss-free, non-linearizable execution (assertion~\ref{th:1}).

\begin{theorem}\label{th:3}
There is no algorithm that implements a terminating read-write register in the partitionable model
that is eventually consistent in all executions.
\end{theorem}

\begin{proof}
Consider an execution $E$ in which no messages are delivered, and in which process $p$ performs
operation $\mathit{write}(p, v)$. This write must terminate, since the algorithm is terminating.
Process $q$ with $p \neq q$ performs $\mathit{read}(q)$ infinitely many times. However, since no
messages are delivered, $q$ can never learn about the value written, so $\mathit{read}(q)$ never
returns $v$. Thus, $E$ is not eventually consistent.
\end{proof}

\subsubsection{Opportunistic properties}

Note that corollary~\ref{corr:2} is about loss-free executions, whereas assertion~\ref{th:3} is
about \emph{all} executions. If we limit ourselves to loss-free executions, then eventual
consistency \emph{is} possible (e.g.\ by maintaining a replica at each process, and broadcasting
every write to all processes).

However, everything we have discussed in this section pertains to the partitionable model, in which
we cannot assume that all executions are loss-free. For clarity, we should specify the properties of
an algorithm such that they hold for \emph{all} admissible executions of a given system model, not
only selected executions.

To this end, we can transform a property $\mathcal{P}$ into an \emph{opportunistic} property
$\mathcal{P}^\prime$ such that:
$$\forall E \mathbin{:} (E \models \mathcal{P}^\prime) \Leftrightarrow (\mathit{lossfree}(E) \Rightarrow (E \models \mathcal{P}))$$
or, equivalently:
$$\forall E \mathbin{:} (E \models \mathcal{P}^\prime) \Leftrightarrow (\mathit{partitioned}(E) \vee (E \models \mathcal{P}))\,.$$
In other words, $\mathcal{P}^\prime$ is trivially satisfied for executions that are partitioned.
Requiring $\mathcal{P}^\prime$ to hold for all executions is equivalent to requiring $\mathcal{P}$
to hold for all loss-free executions.

Hence we define an execution $E$ as \emph{opportunistically eventually consistent} if $E$ is
partitioned or if $E$ is eventually consistent. (This is a weaker liveness property than eventual
consistency.)

Similarly, we define an execution $E$ as \emph{opportunistically terminating linearizable} if $E$ is
partitioned, or if $E$ is linearizable and every operation in $E$ terminates in finite time.

From the results above, we can see that opportunistic terminating linearizability is impossible in
the partitionable model (corollary~\ref{corr:2}), whereas opportunistic eventual consistency is
possible. This distinction can be understood as the key result of CAP. However, it is arguably not a
very interesting or insightful result.

\subsection{Mismatch between formal model and practical systems}

Many of the problems in this section are due to the fact that availability is defined by Gilbert and
Lynch as a liveness property (section~\ref{sec:no-max-latency}). Liveness properties make statements
about something happening \emph{eventually} in an infinite execution, which is confusing to
practitioners, since real systems need to get things done in a finite (and usually short) amount of
time.

Quoting Lamport~\cite{Lamport2000bi}: ``Liveness properties are inherently problematic. The question
of whether a real system satisfies a liveness property is meaningless; it can be answered only by
observing the system for an infinite length of time, and real systems don't run forever. Liveness is
always an approximation to the property we really care about. We want a program to terminate within
100 years, but proving that it does would require the addition of distracting timing assumptions.
So, we prove the weaker condition that the program eventually terminates. This doesn't prove that
the program will terminate within our lifetimes, but it does demonstrate the absence of infinite
loops.''

Brewer~\cite{Brewer2012ba} and some commercial database vendors~\cite{Aerospike2014wa} state that
``all three properties [consistency, availability, and partition tolerance] are more continuous than
binary''. This is in direct contradiction to Gilbert and Lynch's formalization of CAP (and our
restatement thereof), which expresses consistency and availability as safety and liveness properties
of an algorithm, and partitions as a property of the system model. Such properties either hold or
they do not hold; there is no degree of continuity in their definition.

Brewer's informal interpretation of CAP is intuitively appealing, but it is not a theorem, since it
is not expressed formally (and thus cannot be proved or disproved) -- it is, at best, a rule of
thumb. Gilbert and Lynch's formalization can be proved correct, but it does not correspond to
practitioners' intuitions for real systems. This contradiction suggests that although the formal
model may be true, it is not useful.

\section{Alternatives to CAP}\label{sec:alternatives}

In section~\ref{sec:definitions} we explored the definitions of the terms \emph{consistency},
\emph{availability} and \emph{partition tolerance}, and noted that a wide range of ambiguous and
mutually incompatible interpretations have been proposed, leading to widespread confusion. Then,
in section~\ref{sec:proofs} we explored Gilbert and Lynch's definitions and proofs in more detail,
and highlighted some problems with the formalization of CAP.

All of these misunderstandings and ambiguity lead us to asserting that CAP is no longer an
appropriate tool for reasoning about systems. A better framework for describing trade-offs is
required. Such a framework should be simple to understand, match most people's intuitions, and use
definitions that are formal and correct.

In the rest of this paper we develop a first draft of an alternative framework called
\emph{delay-sensitivity}, which provides tools for reasoning about trade-offs between consistency
and robustness to network faults. It is based on to several existing results from the distributed
systems literature (most of which in fact predate CAP).

\subsection{Latency and availability}\label{sec:latency-availability}

As discussed in section~\ref{sec:no-max-latency}, the latency (response time) of operations is often
important in practice, but it is deliberately ignored by Gilbert and Lynch.

The problem with latency is that it is more difficult to model. Latency is influenced by many
factors, especially the delay of packets on the network. Many computer networks (including Ethernet
and the Internet) do not guarantee bounded delay, i.e.\ they allow packets to be delayed
arbitrarily. Latencies and network delays are therefore typically described as probability
distributions.

On the other hand, network delay can model a wide range of faults. In network protocols that
automatically retransmit lost packets (such as TCP), transient packet loss manifests itself to the
application as temporarily increased delay. Even when the period of packet loss exceeds the TCP
connection timeout, application-level protocols often retry failed requests until they succeed, so
the effective latency of the operation is the time from the first attempt until the successful
completion. Even network partitions can be modelled as large packet delays (up to the duration of
the partition), provided that the duration of the partition is finite and lost packets are
retransmitted an unbounded number of times.

Abadi~\cite{Abadi2012hb} argues that there is a trade-off between consistency and latency, which
applies even when there is no network partition, and which is as important as the
consistency/availability trade-off described by CAP. He proposes a ``PACELC'' formulation to reason
about this trade-off.

We go further, and assert that availability should be modeled in terms of operation latency. For
example, we could define the availability of a service as the proportion of requests that meet some
latency bound (e.g.\ returning successfully within 500~ms, as defined by an SLA). This
empirically-founded definition of availability closely matches our intuitive understanding.

We can then reason about a service's tolerance of network problems by analyzing how operation
latency is affected by changes in network delay, and whether this pushes operation latency over the
limit set by the SLA. If a service can sustain low operation latency, even as network delay
increases dramatically, it is more tolerant of network problems than a service whose latency
increases.

\subsection{How operation latency depends on network delay}\label{sec:network-dependence}

To find a replacement for CAP with a latency-centric viewpoint we need to examine how operation
latency is affected by network latency at different levels of consistency. In practice, this depends
on the algorithms and implementation of the particular software being used. However, CAP
demonstrated that there is also interest in theoretical results identifying the fundamental limits
of what can be achieved, regardless of the particular algorithm in use.

Several existing impossibility results establish lower bounds on the operation latency as a function
of the network delay $d$. These results show that any algorithm guaranteeing a particular level of
consistency cannot perform operations faster than some lower bound. We summarize these results in
table~\ref{tab:op-latency} and in the following sections.

Our notation is similar to that used in complexity theory to describe the running time of an
algorithm. However, rather than being a function of the size of input, we describe the latency of an
operation as a function of network delay.

In this section we assume unbounded network delay, and unsynchronized clocks (i.e.\ each process has
access to a clock that progresses monotonically at a rate approximately equal to real time, but the
synchronization error between clocks is unbounded).

\begin{table}
    \centering
    \begin{tabular}{ccc}
        Consistency level      & \textit{write} & \textit{read} \\
                               & latency        & latency \\[3pt] \hline \noalign{\vspace{6pt}}
        linearizability        & $O(d)$         & $O(d)$  \\[3pt]
        sequential consistency & $O(d)$         & $O(1)$  \\[3pt]
        causal consistency     & $O(1)$         & $O(1)$  \\[3pt] \hline
    \end{tabular}
    \caption{Lowest possible operation latency at various consistency levels, as a function of
    network delay $d$.}
    \label{tab:op-latency}
\end{table}

\subsubsection{Linearizability}

Attiya and Welch~\cite{Attiya1994gw} show that any algorithm implementing a linearizable read-write
register must have an operation latency of at least $u/2$, where $u$ is the uncertainty of delay in
the network between replicas.\footnote{Attiya and Welch~\cite{Attiya1994gw} originally proved a
bound of $u/2$ for write operations (assuming two writer processes and one reader), and a bound of
$u/4$ for read operations (two readers, one writer). The $u/2$ bound for read operations is due to
Mavronicolas and Roth~\cite{Mavronicolas1999eb}.}

In this proof, network delay is assumed to be at most $d$ and at least $d-u$, so $u$ is the
difference between the minimum and maximum network delay. In many networks, the maximum possible
delay (due to network congestion or retransmitting lost packets) is much greater than the minimum
possible delay (due to the speed of light), so $u \approx d$. If network delay is unbounded,
operation latency is also unbounded.

For the purposes of this survey, we can simplify the result to say that linearizability requires the
latency of read and write operations to be proportional to the network delay $d$. This is indicated
in table~\ref{tab:op-latency} as $O(d)$ latency for reads and writes. We call these operations
\emph{delay-sensitive}, as their latency is sensitive to changes in network delay.

\subsubsection{Sequential consistency}

Lipton and Sandberg~\cite{Lipton1988uh} show that any algorithm implementing a sequentially
consistent read-write register must have $|r| + |w| \geq d$, where $|r|$ is the latency of a read
operation, $|w|$ is the latency of a write operation, and $d$ is the network delay. Mavronicolas
and Roth~\cite{Mavronicolas1999eb} further develop this result.

This lower bound provides a degree of choice for the application: for example, an application that
performs more reads than writes can reduce the average operation latency by choosing $|r| = 0$ and
$|w| \geq d$, whereas a write-heavy application might choose $|r| \geq d$ and $|w| = 0$. Attiya and
Welch~\cite{Attiya1994gw} describe algorithms for both of these cases (the $|r| = 0$ case is similar
to the Zab algorithm used by Apache ZooKeeper~\cite{Junqueira2011jc}).

Choosing $|r| = 0$ or $|w| = 0$ means the operation can complete without waiting for any network
communication (it may still send messages, but need not wait for a response from other nodes). The
latency of such an operation thus only depends on the local database algorithms: it might be
constant-time $O(1)$, or it might be $O(\log n)$ where $n$ is the size of the database, but either
way it is independent of the network delay $d$, so we call it \emph{delay-independent}.

In table~\ref{tab:op-latency}, sequential consistency is described as having fast reads and slow
writes (constant-time reads, and write latency proportional to network delay), although these roles
can be swapped if an application prefers fast writes and slow reads.

\subsubsection{Causal consistency}

If sequential consistency allows the latency of \emph{some} operations to be independent of network
delay, which level of consistency allows \emph{all} operation latencies to be independent of the
network? Recent results~\cite{Attiya2015dm, Mahajan2011wz} show that causal
consistency~\cite{Ahamad1995gl} with eventual convergence is the strongest possible consistency
guarantee with this property.\footnote{There are a few variants of causal consistency, such as
\emph{real time causal}~\cite{Mahajan2011wz}, \emph{causal+}~\cite{Lloyd2011hz} and
\emph{observable causal}~\cite{Attiya2015dm} consistency. They have subtle differences, but we do
not have space in this paper to compare them in detail.}

Read Your Writes~\cite{Terry1994fp}, PRAM~\cite{Lipton1988uh} and other weak consistency models (all
the way down to eventual consistency, which provides no safety property~\cite{Bailis2013jc}) are
weaker than causal consistency, and thus achievable without waiting for the network.

If tolerance of network delay is the only consideration, causal consistency is the optimal
consistency level. There may be other reasons for choosing weaker consistency levels (for example,
the metadata overhead of tracking causality~\cite{CharronBost1991ec, Attiya2015dm}), but these
trade-offs are outside of the scope of this discussion, as they are also outside the scope of CAP.

\subsection{Heterogeneous delays}

A limitation of the results in section~\ref{sec:network-dependence} is that they assume the
distribution of network delays is the same between every pair of nodes. This assumption is not true
in general: for example, network delay between nodes in the same datacenter is likely to be much
lower than between geographically distributed nodes communicating over WAN links.

If we model network faults as periods of increased network delay
(section~\ref{sec:latency-availability}), then a network partition is a situation in which the delay
between nodes within each partition remains small, while the delay across partitions increases
dramatically (up to the duration of the partition).

For $O(d)$ algorithms, which of these different delays do we need to assume for $d$? The answer
depends on the communication pattern of the algorithm.

\subsubsection{Modeling network topology}

For example, a replication algorithm that uses a single leader or primary node requires all write
requests to contact the primary, and thus $d$ in this case is the network delay between the client
and the leader (possibly via other nodes). In a geographically distributed system, if client and
leader are in different locations, $d$ includes WAN links. If the client is temporarily partitioned
from the leader, $d$ increases up to the duration of the partition.

By contrast, the ABD algorithm~\cite{Attiya1995bm} waits for responses from a majority of replicas,
so $d$ is the largest among the majority of replicas that are fastest to respond. If a minority of
replicas is temporarily partitioned from the client, the operation latency remains independent of
the partition duration.

Another possibility is to treat network delay within the same datacenter, $d_\textit{local}$,
differently from network delay over WAN links, $d_\textit{remote}$, because usually
$d_\textit{local} \ll d_\textit{remote}$. Systems such as COPS~\cite{Lloyd2011hz}, which place a
leader in each datacenter, provide linearizable operations within one datacenter (requiring
$O(d_\textit{local})$ latency), and causal consistency across datacenters (making the request
latency independent of $d_\textit{remote}$).

\subsection{Delay-independent operations}\label{sec:disconnected}

The big-$O$ notation for operation latency ignores constant factors (such as the number of network
round-trips required by an algorithm), but it captures the essence of what we need to know for
building systems that can tolerate network faults: what happens if network delay dramatically
degrades? In a delay-sensitive $O(d)$ algorithm, operation latency may increase to be as large as
the duration of the network interruption (i.e.\ minutes or even hours), whereas a delay-independent
$O(1)$ algorithm remains unaffected.

If the SLA calls for operation latencies that are significantly shorter than the expected duration
of network interruptions, delay-independent algorithms are required. In such algorithms, the time
until replica convergence is still proportional to $d$, but convergence is decoupled from operation
latency. Put another way, delay-independent algorithms support \emph{disconnected} or \emph{offline
operation}. Disconnected operation has long been used in network file systems~\cite{Kistler1992bt}
and automatic teller machines~\cite{Brewer2012tr}.

For example, consider a calendar application running on a mobile device: a user may travel through a
tunnel or to a remote location where there is no cellular network coverage. For a mobile device,
regular network interruptions are expected, and they may last for days. During this time, the user
should still be able to interact with the calendar app, checking their schedule and adding events
(with any changes asynchronously propagated when an internet connection is next available).

However, even in environments with fast and reliable network connectivity, delay-independent
algorithms have been shown to have performance and scalability benefits: in this context, they are
known as \emph{coordination-free}~\cite{Bailis2014th} or \emph{ALPS}~\cite{Lloyd2011hz} systems.
Many popular database integrity constraints can be implemented without synchronous coordination
between replicas~\cite{Bailis2014th}.

\subsection{Proposed terminology}\label{sec:terminology}

Much of the confusion around CAP is due to the ambiguous, counter-intuitive and contradictory
definitions of terms such as \emph{availability}, as discussed in section~\ref{sec:definitions}. In
order to improve the situation and reduce misunderstandings, there is a need to standardize
terminology with simple, formal and correct definitions that match the intuitions of practitioners.

Building upon the observations above and the results cited in section~\ref{sec:network-dependence},
we propose the following definitions as a first draft of a \emph{delay-sensitivity framework} for
reasoning about consistency and availability trade-offs. These definitions are informal and intended
as a starting point for further discussion.

\begin{description}
\item[Availability] is an \emph{empirical metric}, not a property of an algorithm. It is defined as
    the percentage of successful requests (returning a non-error response within a predefined
    latency bound) over some period of system operation.
\item[Delay-sensitive] describes algorithms or operations that need to wait for network
    communication to complete, i.e.\ which have latency proportional to network delay. The opposite
    is \emph{delay-independent}. Systems must specify the nature of the sensitivity (e.g.\ an
    operation may be sensitive to intra-datacenter delay but independent of inter-datacenter delay).
    A fully delay-independent system supports \emph{disconnected (offline) operation}.
\item[Network faults] encompass packet loss (both transient and long-lasting) and unusually large
    packet delay. \emph{Network partitions} are just one particular type of network fault; in most
    cases, systems should plan for all kinds of network fault, and not only partitions. As long as
    lost packets or failed requests are retried, they can be modeled as large network delay.
\item[Fault tolerance] is used in preference to \emph{high availability} or \emph{partition
    tolerance}. The maximum fault that can be tolerated must be specified (e.g.\ ``the algorithm can
    tolerate up to a minority of replicas crashing or disconnecting''), and the description must
    also state what happens if more faults occur than the system can tolerate (e.g.\ all requests
    return an error, or a consistency property is violated).
\item[Consistency] refers to a spectrum of different \emph{consistency models} (including
    linearizability and causal consistency), not one particular consistency model. When a particular
    consistency model such as linearizability is intended, it is referred to by its usual name. The
    term \emph{strong consistency} is vague, and may refer to linearizability, sequential
    consistency or one-copy serializability.
\end{description}


\section{Conclusion}

In this paper we discussed several problems with the CAP theorem: the definitions of consistency,
availability and partition tolerance in the literature are somewhat contradictory and
counter-intuitive, and the distinction that CAP draws between ``strong'' and ``eventual''
consistency models is less clear than widely believed.

CAP has nevertheless been very influential in the design of distributed data systems. It deserves
credit for catalyzing the exploration of the design space of systems with weak consistency
guarantees, e.g.\ in the NoSQL movement. However, we believe that CAP has now reached the end of its
usefulness; we recommend that it should be relegated to the history of distributed systems, and no
longer be used for justifying design decisions.

As an alternative to CAP, we propose a simple \emph{delay-sensitivity} framework for reasoning about
trade-offs between consistency guarantees and tolerance of network faults in a replicated database.
Every operation is categorized as either $O(d)$, if its latency is sensitive to network delay, or
$O(1)$, if it is independent of network delay. On the assumption that lost messages are
retransmitted an unbounded number of times, we can model network faults (including partitions) as
periods of greatly increased delay. The algorithm's sensitivity to network delay determines whether
the system can still meet its service level agreement (SLA) when a network fault occurs.

The actual sensitivity of a system to network delay depends on its implementation, but -- in keeping
with the goal of CAP -- we can prove that certain levels of consistency cannot be achieved without
making operation latency proportional to network delay. These theoretical lower bounds are
summarized in Table~\ref{tab:op-latency}. We have not proved any new results in this paper, but
merely drawn on existing distributed systems research dating mostly from the 1990s (and thus
predating CAP).

For future work, it would be interesting to model the \emph{probability distribution} of latencies
for different concurrency control and replication algorithms (e.g.\ by extending
PBS~\cite{Bailis2012to}), rather than modeling network delay as just a single number $d$. It would
also be interesting to model the network communication topology of distributed algorithms more
explicitly.

We hope that by being more rigorous about the implications of different consistency levels on
performance and fault tolerance, we can encourage designers of distributed data systems to continue
the exploration of the design space. And we also hope that by adopting simple, correct and intuitive
terminology, we can help guide application developers towards the storage technologies that are most
appropriate for their use cases.

\subsection*{Acknowledgements}

Many thanks to Niklas Ekström, Seth Gilbert and Henry Robinson for fruitful discussions around this topic.

{\footnotesize
\bibliographystyle{plainnat}
\bibliography{references}{}}
\end{document}